\documentclass[11pt,  oneside, a4paper]{article}
\usepackage[utf8]{inputenc}
\usepackage[a4paper, total={6.5in,9in}]{geometry}

\usepackage{graphicx} 
\usepackage{bbm}
\usepackage{amsfonts}
\usepackage{enumitem}
\usepackage{fullpage}
\usepackage{amsmath}
\usepackage{hyperref}
\usepackage{mathtools}
\usepackage{lipsum}
\usepackage{blindtext}
\usepackage{titlesec}
\usepackage{amssymb}
\usepackage{amsthm}
\usepackage{rotating}
\usepackage{subcaption}
\usepackage{caption}
\usepackage{svg}
\usepackage{xcolor}
\usepackage{verbatim}
\usepackage{natbib}
\usepackage{cleveref}

\newtheorem{corollary}{Corollary}
\newtheorem{aspt}{Assumption}
\newtheorem{theorem}{Theorem}
\newtheorem{claim}{Claim}
\newtheorem{proposition}{Proposition}
\newtheorem{remark}{Remark}

\title{The Impact of Competition on Outcomes of Score-Based College Admissions}

\usepackage{authblk}

\author[*]{George Bentley}
\author[*]{Diptangshu Sen}
\author[*]{Juba Ziani}

\affil[*]{Georgia Institute of Technology}
\affil[ ]{\texttt{\{gbentley6,dsen30,jziani3\}@gatech.edu}}

\date{\today}

\begin{document}

\maketitle

\begin{abstract}
We study how the design of admissions policies affects the ability of students admitted to universities. In our model, applicants have a multi-dimensional ability from the point of view of the university, which is a combination of a “type” and of “soft skills.” Universities may differ in how they evaluate quality and have differing preferences on type and soft skills. Then, university admissions rely on a single noisy aggregate signal, such as a test score, that may not fully align with the university's preferences, and a university evaluates applicants through the posterior expectations of their preference metric given the observed signal.

Our main results highlight that the design of good admission policies can be counter-intuitive. Under a single university, when holding the number of qualified applicants constant, increasing the usefulness of the signal (by aligning it more closely with the university preferences) leads to a worse type and soft skill for admitted students. Further, a university cannot affect the composition of students that are strong on type versus soft skills by changing their preferences. The picture becomes even more complicated under competition between as few as two universities: self-selection effects among students admitted to both universities can lead to part of the applicant pool switching which university they prefer, even under small changes in the design of the noisy signal. This can, in particular, lead to sudden and non-monotonic loss in the quality of admitted students when changing the alignment between signal and university preferences. Further, a university can get more students by increasing their selectivity. Finally, when admissions rely on separate noisy scores for type and for soft skills, we show that universities that put more emphasis on type (respectively soft skills) end up, counter-intuitively, admitting students with higher soft skills (respectively type).
\end{abstract}


\section{Introduction}\label{sec:intro}

Universities routinely make high-stakes selection decisions using signals that are only partially informative: they do not fully capture students' underlying abilities and may not fully align with a university's preferences. Standardized tests, grades, essays, recommendations, interviews, and extracurricular records all provide partial information about an applicant; they also often do not directly reveal the multidimensional qualities that admissions committees may value. Some universities may value students with stronger technical or mathematical backgrounds, some may prefer students with strong communication or leadership skills, and some may be interested in a balance of both. I.e., different universities may place different weights on different skills.

This leads to the following question: how should a university use admission signals that are only imperfectly aligned with its own preferences? Intuitively, better alignment between the admission metric and the university preferences should improve their outcomes. A university that values academic preparation should benefit from an admission metric that places greater weight on it. While this intuition is useful heuristically, we show that under uncertainty, and when several universities compete for the same students, such simple design principles often break down.

We study these issues through a stylized but informative model of score-based college admissions. Each applicant has two attributes: a ``type'', which represents their level of academic preparation or technical skill, and a ``soft skill'', which captures non-technical attributes such as communication or leadership. A university evaluates applicants according to a weighted combination of these two attributes, which we call a student's ``quality'', but it does not directly observe them. Instead, admissions decisions are based on a noisy signal, or an ``admissions score'', that combines information from both (e.g., a standardized test or an interview score). Each university then forms a posterior expectation of its own quality metric conditional on the observed admissions score and admits applicants whose expected quality exceeds its threshold, following the Bayesian updating framework used in related models of admissions and downstream selection~\citep{corbett1994knowledge,vanlehn1998student,kannan2018downstream,acharya2022wealth}. 

We highlight a number of counterintuitive and non-monotone effects in admissions parameters and preferences. A university may obtain stronger admitted students from a signal that is less aligned with its own preferences; changing a university's preferences may fail to steer the admitted class towards the kind of skills that the university values more; and, under competition, a university may improve both student quality and enrollment by becoming more selective. These effects show that local and heuristic changes to admissions policies can have the opposite of their intended effect once they are considered in their proper context. Our goal is not to prescribe a single optimal admissions policy, but rather to show that the design landscape is delicate and cannot be reduced to high-level intuitive heuristics.

\paragraph{Summary of contributions.}
\begin{enumerate}
\item In Section~\ref{sec:model}, we introduce a Bayesian model of score-based admissions in which students have both a type and soft skill indicator. Universities are Bayesian decision-makers: they make admissions decisions based on the posterior expectation of their own quality metric induced by the noisy signal they observe, such as a test score.
\item In Section~\ref{sec:smsu}, we analyze the case of students applying to a single university. We characterize the average type, average soft skill, and average quality of admitted students as functions of the admissions metric and the university's preferences. We show that, when the count of qualified applicants is held fixed, a university can counterintuitively admit stronger students when the admissions signal is less aligned with its preferences. We also show that changing the university's preference parameter does not change the ratio of average type to average soft skill among admitted students; that is, the university has little leeway to control how balanced admitted students are.
\item We consider competition in Section~\ref{sec:smmu}. Competition between as few as two universities introduces new effects: small changes in the admissions metric or university preferences can shift which university is more selective with respect to test scores. This leads to non-monotonic, discontinuous changes in a university's average ability and enrollment, due to self-selection effects: a university can suddenly become preferred by applicants if it is perceived as more selective, and hence more prestigious. Then, a university can simultaneously improve student quality and enrollment by becoming more selective than its peers, contrary to the common wisdom that greater selectivity leads to smaller pools of qualified admits.
\item Finally, in Section~\ref{sec:mm}, we show that, under multiple admission metrics, the picture becomes even more complicated. When universities use separate noisy metrics for type and soft skill, threshold-based admissions rules can lead to reversals: a university that gives more weight to type may admit students with higher soft skills. In comparison, one that gives more weight to soft skill may admit students with higher types.
\end{enumerate}

\paragraph{Related Work} 
In recent years, a lot of work has focused on how to make high-stakes admissions decisions \textit{efficiently}, \textit{fairly} and \textit{at scale}. A key aspect of this literature involves a mechanism design approach to admissions to schools and colleges, where the focus is on key aspects like efficient matching/assignment with nice properties like respecting preferences, monotonicity, strategy-proofness etc.~\citep{holster2019college,abdulkadirouglu2003school,gong2025dynamic,jiang2019stable,larroucau2022dynamic,biro2022large,rios2021improving,abizada2015stability,kara1997implementation,pathak2017really,feigenbaum2020dynamic,leshno2021cutoff,allman2022designing}.

Our focus, however, is on college admissions through the use of student scores on some well-defined admission metrics. It is relatively well-understood that admission metrics are often imperfect~\citep{minaya2018labor}, different metrics have different levels of screening or sorting power~\citep{gandil2022college,verostek2021analyzing} and use of such metrics can have disparate impacts on under-privileged populations~\citep{green1980impact,fleming1998standardized,au2022unequal,bajwa2023test}. However, despite the scrutiny, such metrics continue to endure as useful evaluation tools for universities because of their simplicity and scalability.   

The main objective of our work is to study the impact of competition in score-based college admissions where universities compete with each other to recruit the best students. Prior work has studied different aspects of how universities may try to earn a competitive edge in this game, for example, through appropriate design of tuition and admission policies~\citep{fu2014equilibrium,shulruf2015admission} or by offering early admits~\citep{antecol2012early}. While these works study interventions for universities, we try to understand how within the existing framework of score-based admission policies. university-side competition may lead to fundamentally different admission outcomes (in terms of the average type and soft skill of admitted students). Existing work also considers competition on the student-side as they compete to enroll into selective universities with limited capacities~\citep{jurajda2011gender,pastine2012student,macleod2015reputation}, but this line of work is not directly related to our paper. 

We study the social impact of admission decisions, drawing a parallel with the literature on fairness impacts of selection or admission rules on process outcomes---e.g.,test-optional admissions~\citep{liu2021test,zhang2025rethinking,sirolly2025impact} or affirmative action~\cite{garg2020dropping,garg2021standardized,lee2024ending}. Key similarities include: i) presence of information gaps for the decision-maker and imperfect evaluations~\citep{liu2021test,garg2020dropping,fish2026}; and ii) using Bayesian updates to estimate applicant's true ``quality" given noisy/imperfect scores~\citep{emelianov2022fairness,acharya2022wealth,zhang2025rethinking,kannan2018downstream,liu2021test}. The fairness literature has also considered different kinds of composition effects and how they can lead to disparate outcomes at the individual and group level~\citep{dwork2018individual,dwork2018group,dwork2020individual,arunachaleswaran2022pipeline,gradwohl2026fairness}. Our results are similar in flavor, but in the context of university-side competition and how admission outcomes can be highly counterintuitive.

\section{Model}\label{sec:model}
We consider a setting where there are $N$ universities $\{\mathcal{U}_i \}_{i \in [N]}$ and a large population of students who apply to get admitted to these universities. Throughout the paper, we consider either $N=1$ (single university case) or $N=2$ (multi-university case) --- in fact, we show that competition effects are already non-trivial with just $2$ universities. 

A student is characterized by a tuple $(t, s)$---$t$ is their ``type" (indicating academic prowess) and $s$ is their ``soft skill" level. We make the following assumption on $t$ and $s$: 

\begin{aspt}\label{aspt:soft_type_normal}
Each student's type $t$ and soft skill $s$ are drawn IID from two independent standard normal distributions, i.e., $s \sim \mathcal{N}(0,1)$, $t \sim \mathcal{N}(0, 1)$, and $s \perp t$.  
\end{aspt}
\begin{remark}
We note that Gaussian assumptions to model populations are fairly common in the literature~\citep{kannan2018downstream,garg2020dropping,liu2021test,fish2026}. While the distributions of $t$ and $s$ in the population can potentially have different spreads (different variances), they can always be normalized and therefore, we use a variance of $1$ for each without loss of generality. 
\end{remark}

\paragraph{Admission Process with a Single Metric.} Students are admitted based on their scores on a \textit{noisy} and \textit{imperfect} admission metric $A$ (used by all universities) where $A \triangleq \alpha t + (1-\alpha)s + \sigma_Z Z$. Here, $\alpha \in (0, 1)$, therefore $A$ measures a student's preparedness for college using a convex combination of their true type and soft skill. $Z \sim \mathcal{N}(0, 1)$ captures the random noise in each measurement with $\sigma_Z$ being the level of noise. Note that $\sigma_Z$ determines how informative the signal is, i.e., a higher level of noise diminishes the usefulness of the signal in accurately evaluating students. As $\sigma_Z \to 0$, the signal becomes an exact measure of the target admission metric $\alpha t + (1-\alpha)s$, while when $\sigma_Z \to +\infty$, the signal becomes completely uninformative. Such noisy observation models have been used extensively in the literature~\citep{acharya2022wealth,kannan2018downstream}. 

Each university $\mathcal{U}_i$ measures student quality using their own metric $Q_i \triangleq \beta_i t + (1-\beta_i)s$ where $\beta_i \in (0, 1)$ is their preference parameter capturing how much they prioritize student type compared to soft skill. Note that $\beta_i$ may be very different from $\alpha$ (how the admission metric $A$ weights type and soft skill) and this drives one of the fundamental tensions in our paper --- each university has a fundamentally different notion of student quality which may be misaligned with the admission metric they need to use for admission. University $\mathcal{U}_i$ has a \textit{quality threshold} $\tau_{Q,i}$, i.e., a student with score $a$ on the admission metric $A$ is considered \emph{qualified} with respect to $\mathcal{U}_i$ if and only if:
\begin{align}
Q_i \geq \tau_{Q,i}.
\end{align}
However, the university can only observe $A$, and cannot observe $Q$ directly. Instead, the university makes decisions on their posterior distribution on the quality of a student given $A$. I.e., the university \emph{admits} a student if and only if:
\begin{align}\label{eq:quality_threshold}
    \mathbb{E}[Q_i ~|~A = a] \geq \tau_{Q,i}. 
\end{align}

\begin{remark}
In much of the paper, we focus on the case where $\tau_{Q,i} > 0$. The average quality of a student in our model is (without loss of generality) $0$. This means that our assumption corresponds, in practice, to a selective university that only is interested in admitting above-average students. 
\end{remark}

We now show that each university's quality threshold $\tau_{Q,i}$ (which is not enforceable because it is not directly observed) always maps to a threshold $\tau_{A, i}$ on the admission metric $A$. We call $\tau_{A, i}$ university $\mathcal{U}_i$'s \textit{admission threshold}. Formally, 
\begin{claim}\label{clm:adm_threshold}
Let $\sigma_{\alpha} = \sqrt{\alpha^2 + (1-\alpha)^2 + \sigma_Z^2}$. Then for university $\mathcal{U}_i$ with preference parameter $\beta_i$ and quality threshold $\tau_{Q,i}$, setting an admission threshold of $\tau_{A, i} = \frac{\tau_{Q,i} \sigma_{\alpha}^2 }{\alpha \beta_i + (1-\alpha)(1-\beta_i)}$ guarantees that each admitted student with score $a \geq \tau_{A_i}$ on metric $A$ must also have expected quality $Q_i$ at least as large as $\tau_{Q,i}$.  
\end{claim}
The proof can be found in Appendix~\ref{app:adm_threshold}. The main idea here is that given a student's score on admission metric $A$, the university can perform a Bayesian update and determine the expected type and soft skill of the student as follows:
\[
    t~|~(A = a) \sim \mathcal{N}\left(\frac{\alpha a}{\sigma_{\alpha}^2}, \frac{(1-\alpha)^2 + \sigma_Z^2}{\sigma_{\alpha}^2} \right);
\]
\[
    s~|~(A = a) \sim \mathcal{N}\left( \frac{(1-\alpha) a}{\sigma_{\alpha}^2}, \frac{\alpha^2 + \sigma_Z^2}{\sigma_{\alpha}^2}\right), 
\]
thereby also enabling them to estimate the ``quality" of the said student (according to their own definition of quality).

\begin{remark}
Importantly, $Q_i$ is a metric that captures the university's preferences on students. The admission score $A$ is the \emph{physical object} on which the university makes decisions, in order to satisfy their preference $Q_i$. Our result shows how to interchangeably map one to the other.   
\end{remark}

\paragraph{Renormalized quality} It is often useful to think about the renormalized quality when thinking of the impact of $\beta_i$ on the parameters of the problem. In particular, note that changing $\beta_i$ in fact changes the variance of the non-normalized metric $Q_i$; then, the same level of quality across different $\beta_i$ can lead to different ranking or quantile that any given student corresponds to within the population. In turn, it is natural to define the normalized quality
\begin{align}
Q_{norm,i} = \frac{\beta_i t + (1-\beta_i) s}{\sigma_\beta},
\end{align}
where $\sigma_\beta = \sqrt{\beta_i^2 + (1-\beta_i)^2}$. We note that quality and renormalized quality are interchangeable when $\beta_i$ is constant, and refer to ``quality'' for simplicity of exposition; results when $\beta_i$ vary are for \emph{renormalized} quality. 

\paragraph{Multiple Metrics.} Some admission processes may use multiple noisy metrics $A_1, A_2,...A_k$. This introduces new complexities into the design of the process because each university now has to interpret multiple signals about a student (with potential overlap) to make decisions. In particular, the conditional distribution of $Q_i~|~(A_1 = a_1,..A_k = a_k)$ may be complex to characterize, even under Gaussian assumptions. 
We expand on some tractable special cases of this model in Section~\ref{sec:mm}.

\section{The Case of a Single University}\label{sec:smsu}

We first consider a simple setting where there is a single university that uses metric $A$ to make admission decisions on students. In this setting, we assume that all admitted students will choose to attend the university.
Note that we slightly abuse notation here and drop all subscripts $i$ (because there is one university) for simplicity of exposition. This is a warm-up section where we study admission outcomes in the absence of competition, which will subsequently become our baseline of comparison in Sections~\ref{sec:smmu} and~\ref{sec:mm}.

Our first result characterizes the admission outcomes in the single university case in terms of the average type and average soft skill of admitted students: 

\begin{proposition} \label{prop:meantypesoft}
The average type and average soft skill level of an admitted student are given by $\frac{\alpha}{\sigma_{\alpha}}\cdot H(\tau_A/\sigma_{\alpha})$ and $\frac{(1-\alpha)}{\sigma_{\alpha}}\cdot H(\tau_A/\sigma_{\alpha})$, respectively, where $H(\cdot)$ represents the hazard rate of the standard normal random variable. Further, the average quality of an admitted student is given by 
\begin{align*}
\frac{\beta \alpha + (1-\beta) (1-\alpha) }{\sigma_\alpha} \cdot H(\tau_A/\sigma_{\alpha}).
\end{align*}
\end{proposition}
\noindent 
The proof can be found in Appendix~\ref{app:meantypesoft}. Further, we capture the number of students, or \emph{fraction of the total student pool of applicants} that are qualified below: 

\begin{proposition}\label{prop:quantile}
    Let $\sigma_\beta=\sqrt{\beta^2+(1-\beta)^2}$ be the standard deviation of the quality $Q$, and $\Phi(\cdot)$ be the cdf of a standard Normal variable. The fraction of students that satisfy a threshold on preferences $\tau_Q$ is given by $1-\Phi(\tau_Q/\sigma_\beta)$. In turn, a university who wants to set $\tau_Q$ such that only the top $\delta$-fraction of students is qualified, according to their preferences, must set $\tau_Q=(\Phi^{-1}(1-\delta))\cdot \sigma_\beta$.
\end{proposition}
\noindent 
This result is a direct corollary of Proposition~\ref{prop:meantypesoft}. Note that the threshold to keep the fraction of the qualified population constant is affected by $\sigma_\beta$ hence $\beta$, but is independent of $\alpha$.

The rest of the section parses this technical result to provide insights on how the admission outcomes depend on key parameters of interest like $\alpha$, $\beta$, $\tau_A$ and $\sigma_Z$. 

\paragraph{Impact of the admission thresholds $\tau_Q$, $\tau_A$:} As per intuition, as the university increases their standards for admission through a higher threshold $\tau_Q$ (thereby, by \Cref{clm:adm_threshold}, linearly increasing the threshold $\tau_A$ along metric $A$), the expected quality of admitted students also increases. 

\begin{corollary}\label{corr:smsu_tau}
    At fixed $\alpha$, the average type and soft skill of students admitted to the university are monotonically increasing functions of $\tau_Q$ and $\tau_A$.  Therefore, the average quality of admitted students is also increasing in $\tau_Q$ and $\tau_A$.  
    \label{monohaz}
\end{corollary}
\noindent
This is a direct consequence of Proposition~\ref{prop:meantypesoft} and the monotonicity of the hazard rate function of the standard normal.  However, increasing the admission threshold does adversely affect the number of students who can be admitted:

\begin{corollary}\label{corr:mononumstu}
    As the admission thresholds $\tau_Q$ and $\tau_A$ increase, the number of students admitted to the university decreases.
\end{corollary}
\noindent 
This highlights an interesting trade-off introduced by selectivity: a high admission threshold means that only a limited number of students are sufficiently ``qualified" to be admitted, but those who get admitted are of high average quality.

\paragraph{Impact of the noise level $\sigma_Z$:} 
We now examine how the average type and soft skill of admitted students are influenced by the level of noise $\sigma_Z$. 

\begin{proposition}\label{prop:noiselim}
As $\sigma_Z$ increases, for $\tau_Q > 0,$ the mean type and mean soft skill at a university monotonically decrease, and the number of students also monotonically decreases.  Furthermore, as $\sigma_Z\rightarrow\infty$, the mean type approaches $\frac{\alpha\tau_Q}{\alpha\beta+(1-\alpha)(1-\beta)}$, and the mean soft skill level approaches $\frac{(1-\alpha)\tau_Q}{\alpha\beta+(1-\alpha)(1-\beta)}$.
\end{proposition}
\noindent 
The proof is in Appendix~\ref{app:noiselim}. Unsurprisingly, an increased signal quality leads to better outcomes for the university. 

\paragraph{Impact of $\alpha$:} Here, we consider the impact of $\alpha$ on the average quality of admitted students: as $\alpha$ increases, and the admission metric $A$ puts more importance on type compared to soft skill, the expected type of admitted students increases while at the same time the expected soft skill of admitted students decreases. Formally:

\begin{proposition}\label{prop:alphamon}
    The average type of admitted students monotonically increases in $\alpha$, while the average soft skill of admitted students monotonically decreases in $\alpha$.
\end{proposition}
\noindent 
The proof can be found in Appendix~\ref{app:alphamon}. However, the average quality of an admitted student depends on both type and soft skill which (as we show above) move in opposite directions as $\alpha$ changes. Our next result shows that the average quality of admits is no longer monotonic in $\alpha$; in fact, when there is no noise and $\sigma_Z = 0$, it achieves its maximum value when $\alpha = \beta$; i.e., when the university's measure of quality and the admission metric are perfectly aligned. The same holds true for the number of admits.
\begin{proposition}
    Assuming $\tau_Q > 0$, as we vary $\alpha$ (fixing all other parameters), the average quality of students and the number of students admitted to university $\mathcal{U}$ are both highest when $\alpha=\beta$ under $\sigma_Z = 0$.  Both quantities monotonically increase as $\alpha$ increases from $0$ to $\beta$ and monotonically decrease as $\alpha$ increases from $\beta$ to $1$.
    \label{alphamax}
\end{proposition}
\noindent 
The proof is relegated to Appendix~\ref{app:alphamax}. The intuition is that, as $\beta$ moves away from $\alpha$, finding qualified students becomes more difficult due to the loss of informativeness of the signal\footnote{Similar effects have been observed previously in~\citet{fish2026} in the context of hiring processes where finding applicants qualified along two orthogonal dimensions (in their case, ability and fit) is more difficult.}. In turn, the university overcompensates by hiring more selectively and increase the threshold $\tau_A$ on admission scores in order to balance out the increased risk in decision-making in a less aligned and less informative signal.  Furthermore, because the expected quality of students is Gaussian, and thus the probability density decays super-exponentially, the university setting a threshold farther in the right tail means that most of the students who \textit{do} meet the increased threshold will be close to the threshold. 

\paragraph{Impact of university's preference parameter $\beta$:} The impact of $\beta$ on university admissions is where the picture starts becoming more surprising. Remember that as per \Cref{prop:quantile}, a university interested in improving the quality of admitted students \emph{while keeping the number of qualified students constant}, should set their threshold $\tau_Q =(\Phi^{-1}(1-\delta))\cdot \sigma_\beta$, where $\delta$ is the fraction of the population they want to be qualified\footnote{This can be seen as a university deciding their admission threshold based on known population statistics, and making sure that at least a certain fraction of the population satisfies their desired qualifications.}. This has important implications. Since selectivity corresponds to the proportion of students who meet the university's criterion, we will assume in the forthcoming analysis that $\delta$ is a constant, and thus hold $\tau_Q/\sigma_\beta$ constant.

There, we see that perhaps unsurprisingly, universities get higher quality students when $\beta=\alpha$, i.e. when the signal $A$ is perfectly aligned with the \emph{renormalized} quality metric $Q_{norm}$. As before, and as $\beta$ moves away from $\alpha$, the quality of admitted students with respect to $\beta$ decreases:
\begin{proposition}\label{prop:opdir}
    For fixed $\tau_Q/\sigma_\beta (> 0)$, when varying $\beta$, the average normalized quality $Q_{norm}$ of admitted students, as well as and the number of admitted students, are both highest when $\beta=\alpha$.  Both quantities monotonically increase as $\beta$ increases from $0$ to $\alpha$ and monotonically decrease as $\beta$ increases from $\alpha$ to $1$.
\end{proposition}

The proof is given in Appendix~\ref{app:opdir}. We note that this result is \emph{different} from Proposition~\ref{prop:alphamon}, which holds $\beta$ fixed while varying $\alpha$---in our work, the roles of $\alpha$ and $\beta$ are asymmetric, as made evident by Proposition~\ref{prop:meantypesoft}.  In fact, this distinction creates a crucial difference: as $\beta$ changes, the shape of the distribution of admitted student quality also changes. This, in turn, leads to a decrease, at a desired quality level $\tau_Q$, in the number of admitted students. 

However, a central counter-intuitive insight is that while the university prefers signals $A$ that align \emph{more} with their own preferences $Q$, the average type and average soft skills of students \emph{both} decrease. Formally:

\begin{theorem}\label{thm:betacurve}
For fixed $\tau_Q/\sigma_\beta (> 0)$ and $\alpha$, the average type and soft skill for admitted students are lowest when $\beta=\alpha$ and monotonically increase the farther $\beta$ gets from $\alpha$.
\end{theorem}
\noindent 
The proof is provided in Appendix~\ref{app:betacurve}. The counter-intuitive result is due to the following effect: while the average type and soft skills both decrease with a more informative signal as shown by Theorem~\ref{thm:betacurve}, Proposition~\ref{prop:opdir} still holds because the average quality of admitted students itself increases when $\beta$ is closer to $\alpha$, as the university works with a more informative signal and is able to make better admission decisions.

\paragraph{A university cannot steer their application pool by changing their admission metric} Finally, another counter-intuitive insight comes when universities may change their preferences on type versus soft skill, controlled by $\beta$. A naive designer may assume that changing $\beta$ would steer the population of admitted students towards being more type- versus soft-skills- competent, and $\beta$ could be used as a design parameter used to balance how good students are across these two different dimensions. However, we show that changing $\beta$ has \emph{no effect on how admitted students are balanced across type and soft skill competency.} Formally: 

\begin{corollary}\label{corr:constantrat}
For a fixed metric $A$ (i.e., fixed $\alpha$) used for admission decisions, the ratio of the average type to the average soft skill level of an admitted student depends on $\alpha$, but is completely independent of the university's quality preference $\beta$ or the admission threshold they use.  
\end{corollary}

\section{How Does Competition between Universities Change the Picture?}\label{sec:smmu}

We now consider a setting with two universities, i.e., $N = 2$. Universities $\mathcal{U}_1$ and $\mathcal{U}_2$ compete to admit students from a common applicant pool. Each university uses a similar admission policy as described previously in Section~\ref{sec:model}, where they evaluate applicants on the basis of their scores on noisy metric $A = \alpha t + (1-\alpha)s + \sigma_Z Z$. We remind the reader that for each $i \in \{1, 2\}$, university $\mathcal{U}_i$ has a preference parameter $\beta_i$, which determines their quality metric $Q_i$, and a quality threshold $\tau_{Q,i}$. A student with score $a$ on metric $A$ is offered an admission slot by university $\mathcal{U}_i$ if and only if $\mathbb{E}[Q_i ~|~A = a] \geq \tau_{Q,i}$. 
In this section, we demonstrate that when universities are competing for the same students, admission outcomes across universities can become more complex due to student self-selection. 

\paragraph{University Selectivity:} Recall that as per Claim~\ref{clm:adm_threshold}, each university $\mathcal{U}_i$'s quality threshold $\tau_{Q,i}$ maps to a unique admission threshold $\tau_{A, i}$ on the metric $A$---we showed in particular that $\tau_{A, i} = \frac{\tau_{Q,i} \sigma_{\alpha}^2}{\alpha \beta_i + (1-\alpha)(1-\beta_i)}$. Therefore, we call university $\mathcal{U}_1$ \textit{more selective} if and only if $\tau_{A, 1} > \tau_{A, 2}$ and vice-versa. 
In the special case where $\tau_{A, 1} = \tau_{A, 2}$, both universities are \textit{equally selective}. 

\paragraph{Which University is More Selective?} We now characterize when $\mathcal{U}_1$ is more selective versus when $\mathcal{U}_2$ is more selective, as a function of the admission rule (specifically $\alpha$), the university preference parameters $\beta_1$, $\beta_2$ and their quality thresholds $\tau_{Q,1}, \tau_{Q,2}$. A core insight is that which university is more selective can vary discontinuously with small changes in the above parameters.

To do so, we first tackle some trivial special cases where the more selective university can be determined easily: 

\begin{enumerate}[label=\roman*)]
\item When $\beta_1 = \beta_2$, the order of selectivity is decided directly by the order of the thresholds $\tau_{Q,1}, \tau_{Q,2}$. 
\item If $\tau_{Q,1} = \tau_{Q,2} = 0$, both universities are equally selective, independently of everything else.  
\end{enumerate}
Our next result characterizes the \textit{more selective} university in the general case where the universities have different preference parameters $\beta_1, \beta_2$ and the quality thresholds $\tau_{Q,1}, \tau_{Q,2}$ are both not zero: 

\begin{theorem}\label{thm:more_selective}
WLOG, suppose that $\beta_1 > \beta_2$. Additionally, consider the setting where $\tau_{Q,1},\tau_{Q,2} > 0$ and define $\alpha_{th} = \frac{\tau_{Q,2}-\tau_{Q,1}+\tau_{Q,1}\beta_2-\tau_{Q,2}\beta_1}{\tau_{Q,2}-\tau_{Q,1}+2\tau_{Q,1}\beta_2-2\tau_{Q,2}\beta_1}$. Then,
  \begin{itemize}
      \item if $\frac{\tau_{Q,1} }{\tau_{Q,2} } > \frac{\beta_1}{\beta_2}$, $\mathcal{U}_1$ is more selective (for all $\alpha$);
      \item if $\frac{\tau_{Q,1} }{\tau_{Q,2} } < \frac{1-\beta_1}{1-\beta_2}$, $\mathcal{U}_2$ is more selective (for all $\alpha$);
      \item if $\frac{1-\beta_1}{1-\beta_2}\leq\frac{\tau_{Q,1} }{\tau_{Q,2} }\leq\frac{\beta_1}{\beta_2},$ then $\mathcal{U}_1$ is more selective when $\alpha < \alpha_{th}$, $\mathcal{U}_2$ is more selective when $\alpha > \alpha_{th}$ and both universities are equally selective when $\alpha = \alpha_{th}$. 
  \end{itemize}  
\end{theorem}
The detailed proof of Theorem \ref{thm:more_selective} can be found in Appendix~\ref{app:more_selective}. Since we now have a clear understanding of the selectivity of universities, the rest of the section will be dedicated to studying the admission outcomes of the more selective and less selective universities under competition, where students admitted to both universities make decisions according to either of the two choice models described below.

\paragraph{Student Choice Models:} Note trivially that every student admitted to the more selective university also gets admitted to the less selective university and therefore, faces the choice of deciding which university to attend. We describe below two models of how students make this decision. 
\begin{itemize}
    \item \textit{Random Choice Model.} We consider that on getting admitted to both universities, a student chooses to attend the \textit{more selective} university with probability $q > 0.5$. Which university is more selective is determined in terms of admissions scores, as described above, and is \emph{affected by the choices of $\beta_i$ and $\tau_{Q,i}$}. 
    \item \textit{Self-Sorting Model.} In this model, we consider that on getting admits to both universities, a student always chooses to attend the university $i^*$ at which they are at a higher quantile in terms of quality, \emph{conditional on their type and soft skill}. I.e., a student with type $t$ and soft skill $s$ picks
    \[
        i^* = \arg\max_{i \in \{1, 2\} } \frac{\beta_i t + (1-\beta_i) s}{\sigma_{\beta_i}}.
    \] 
    This corresponds to selectivity not on the admission threshold, but instead in the \emph{number of students that are qualified for each university}, which is controlled by $\frac{\beta_i t + (1-\beta_i) s}{\sigma_{\beta_i}}$ as per Proposition~\ref{prop:quantile}. For simplicity, we assume that ties are broken arbitrarily.
\end{itemize}

If a student that is admitted at both universities decide to accept the offer from University $i$, we say that the student was \emph{admitted} at both universities and \emph{attended} or \emph{enrolled into} University $i$.

\subsection{Competition under the Random Choice Model} 

\paragraph{Impact of Competition on Student Quality} In the presence of a \textit{more selective} university, there are $3$ groups of students: i) a high-scoring group admitted to both universities, ii) a moderate-scoring group only admitted to the less selective university, and iii) a low-scoring group that gets no admits. The high-scoring group decides which university to attend, choosing the more selective university at rate $q$. Thus, the more selective university only receives a subset of the high-scoring students, while the less selective university receives both high and moderate-scoring students. In turn, the more selective university has higher average type and soft skill among attending students compared to the less selective university. We formalize this intuition below: 

\begin{proposition}
    \label{identmorsel}
    In the two-university case, under the random choice model, the more selective university always has the same average type and soft skill as it would in the single-university case.  The less selective university always has strictly lower average type and soft skill than it would in the single-university case.
\end{proposition}
\noindent 
The proof can be found in Appendix \ref{app:identmorsel}. This immediately leads to Corollary \ref{corr:seldom} which says that the more selective university (with the higher $\tau_{A,i}$) must also have higher average type and soft skill of attending students:  

\begin{corollary}\label{corr:seldom}
    Under the \textit{random choice} model, the more selective university always has students with the higher average type \textit{and} the higher average soft skill.
\end{corollary}

\paragraph{Impact of $\sigma_Z$:} We now consider the impact of the amount of noise in the admission signal $A$.
\begin{proposition}
    For fixed $\alpha$, $\beta_i$'s, and $\tau_{Q,i}$'s, which university is more selective (and thus has higher average type and soft skill of attending students) is invariant in the noise $\sigma_Z$.
    \label{kinv}
\end{proposition}
\noindent 
We provide a short proof in Appendix~\ref{app:kinv}. This result should be intuitive; the noise in the admission scores affects both universities equally and does not impact selectivity. In particular, this will imply that none of the insights of the rest of this section are affected by the noise level.

\paragraph{Impact of thresholds $\tau_{Q,i}$ and $\tau_{A,i}$:} We now study the impact of changing the quality threshold $\tau_{Q,i}$ on admissions.
\begin{proposition}
    As $\tau_{Q,i}$ increases for university $\mathcal{U}_i$, the average type and soft skill of students at the university both increase monotonically.
    \label{monotau}
\end{proposition}
\noindent 
The proof can be found in Appendix~\ref{app:monotau}. The result is intuitive: increased selectivity cannot hurt a university in terms of quality of attending students. One issue, however, is that this does not say anything about the number of attending students. In isolation, increasing selectivity decreases the pool of potential qualified candidates. However, and perhaps surprisingly, we show that this is not the case under competition and self-selection effects, as switching from being the less selective to the more selective university can vastly increase one's applicant pool. In particular:

\begin{figure*}[t]
    \centering
    \includegraphics[width=0.9\textwidth]{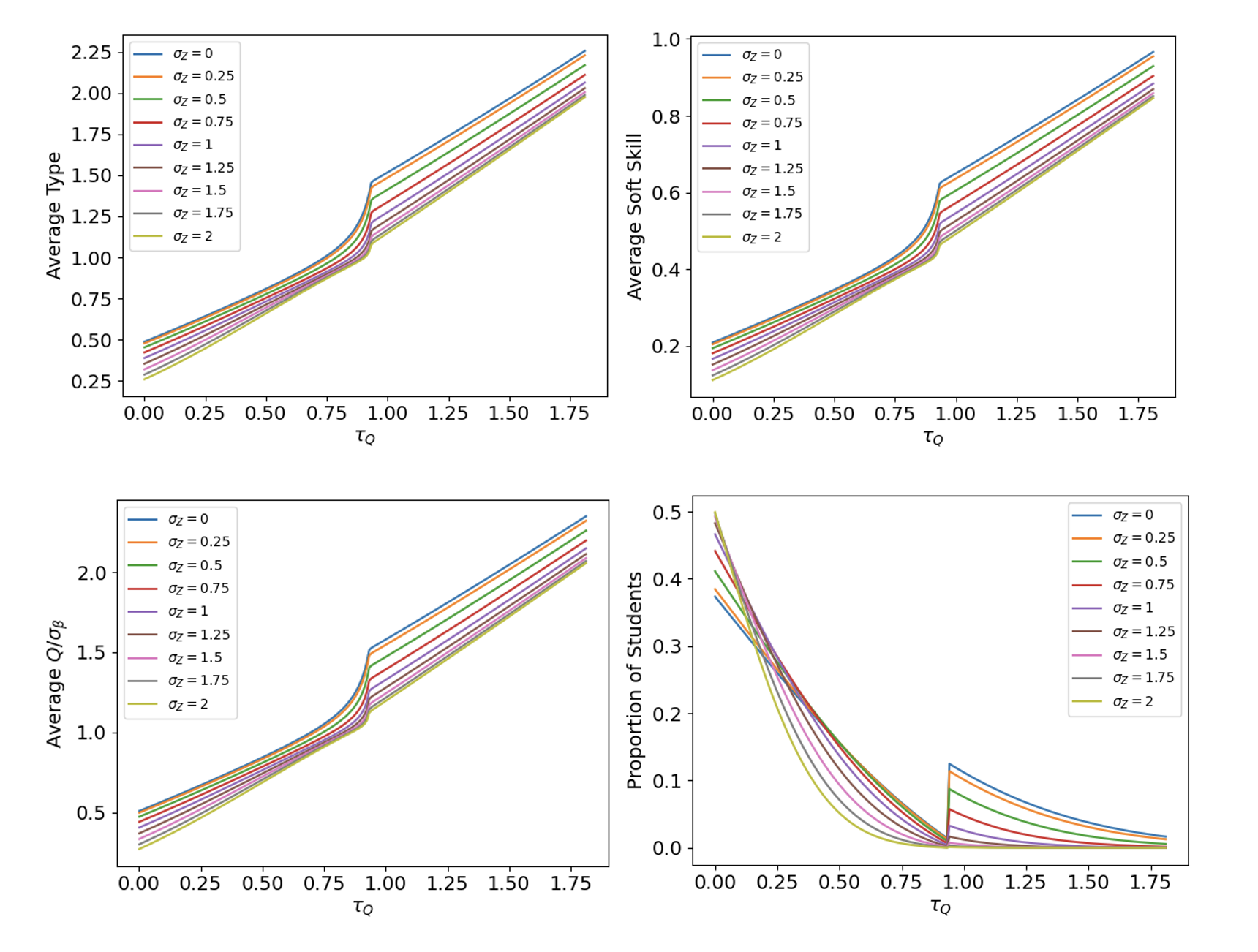}
    \caption{Parameter choice: $\beta_1=0.9$, $\beta_2=0.5$, $\alpha=0.7$, $q=0.9$, and $\tau_{Q,2}/\sigma_{\beta,2}=1$. Average type (top-left), average soft skill (top-right), average quality score (bottom-left), and total proportion of students attending (bottom-right) for University 1. Each colored line represents a different level of noise $\sigma_\alpha$ on the admission rule.}
    \label{fig:tauquad}
\end{figure*}

\begin{theorem}\label{thm:pareto} 
Consider $\tau_{A,1} = \tau_{A,2}-\epsilon$ for some sufficiently small $\epsilon$, i.e. University $\mathcal{U}_2$ is slightly more selective. Then, University $\mathcal{U}_1$ can simultaneously improve i) the average quality of students that choose to attend $\mathcal{U}_1$ and ii) the number of students that choose to attend $\mathcal{U}_1$, by setting $\tau_{A,1} = \tau_{A,2} + \epsilon$. 
\end{theorem}
\noindent 
The proof is given in Appendix~\ref{app:pareto}. Theorem~\ref{thm:pareto} is illustrated by Figure \ref{fig:tauquad}, which shows what happens to the average type, average soft skill, average quality, and proportion of total students at University $\mathcal{U}_1$ for $\beta_1=0.9$, $\beta_2=0.5$, $\alpha=0.7$, $q=0.9$, and $\tau_{Q,2}/\sigma_{\beta_2}=1$.

\begin{figure*}[t]
    \centering
    \includegraphics[width=0.9\textwidth]{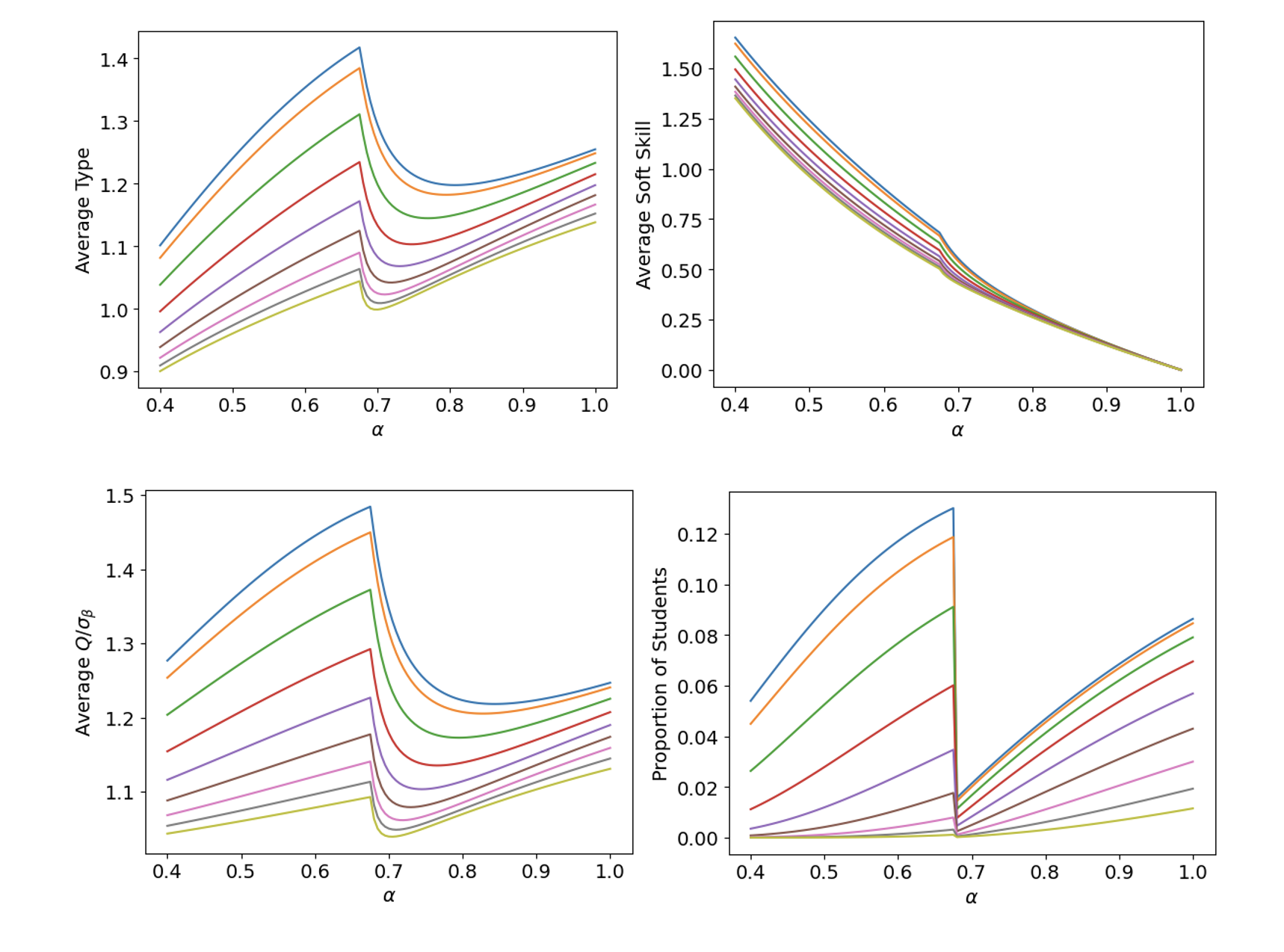}
    \caption{Parameter choice: $\beta_1=0.9$, $\beta_2=0.5$, $q=0.9$, and $\tau_{Q,1}/\sigma_{\beta,1}=\tau_{Q,2}/\sigma_{\beta,2}=1$. Average type (top-left), average soft skill (top-right), average quality score (bottom-left), and total proportion of students attending (bottom-right) for University 1. The lines represent, from top to bottom, $\sigma_Z=0, 0.25, 0.5, 0.75, 1, 1.25, 1.5, 1.75, 2$.}
    \label{fig:quadgraph_alpha}
\end{figure*}

\paragraph{Impact of $\alpha$:} Interestingly, the monotonicity of average type and soft skill of \emph{all} admitted students in $\alpha$ is still true when there are multiple universities:

\begin{proposition}\label{prop:monorange}
    For $\tau_{Q,1}, \tau_{Q,2} \geq 0$,
    the average type monotonically increases in $\alpha$ both for the high-scoring group of students accepted to both universities and the moderate-scoring group accepted to one university.
\end{proposition}
\noindent 
The proof is provided in Appendix~\ref{app:monorange}. This implies that, for each university $\mathcal{U}_i$, any non-monotonicity in average type or soft skill as a function of $\alpha$ is purely the result of Simpson's paradox\footnote{Simpson's paradox is a phenomenon in probability and statistics in which a trend appears in several groups of data but disappears or reverses when the groups are combined.} and the self-selection effects described above.  In particular, for the less selective university, as the thresholds $\tau_{A,1}$ and $\tau_{A,2}$ for $A$ become further apart, the proportion of students who are only accepted to one university increases, putting downward pressure on both the average type and the average soft skill at the less selective university. 

Further, per Theorem~\ref{thm:more_selective}, changing $\alpha$ can also change which university is more selective. Which university is more selective shifts at $\alpha = \alpha_{th}$, which impacts the type, soft skill, and quality of students admitted to both universities.

\begin{theorem}
    Suppose University $\mathcal{U}_1$ is more selective for $\alpha < \alpha_{th}$ and less selective for $\alpha > \alpha_{th}$. When students pick the more selective university with probability $q = 1$, the average type and soft skill at $\mathcal{U}_1$ is discontinuous and non-monotone at $\alpha_{th}$.
    \label{discontalpha}
\end{theorem}
\noindent 
The proof is given in Appendix~\ref{app:discontalpha}. The sharp discontinuity is due to the fact that when University $1$ suddenly becomes less selective, it loses the entire pool of high-scoring students admitted to both universities. However, when $q < 1$, the university that becomes less selective beyond $\alpha = \alpha_{th}$ is still able to retain a fraction $(1-q)$ of the high-scoring students admitted to both universities, which leads to non-monotonicities but preserves continuity. Formally:

\begin{theorem}\label{prop:nonmonalpha}
    Suppose that high-scoring students who get admitted to both universities pick the more selective university with probability $q=1-\epsilon$ for small $\epsilon$. Further, suppose that University $\mathcal{U}_1$ is the more selective university when $\alpha < \alpha_{th}$ and the less selective university with $\alpha > \alpha_{th}$. Then the average type and average soft skill among attending students at $\mathcal{U}_1$ is continuous but non-monotone at $\alpha=\alpha_{th}$. 
\end{theorem}
\noindent 
The proof is in Appendix~\ref{app:nonmonalpha}. For a graphical illustration of the impact of $\alpha$ on admission outcomes, see Figure~\ref{fig:quadgraph_alpha}.

This reversal can happen when changing $\beta$ as well as changing $\alpha$, so Proposition $\ref{prop:opdir}$ does not generalize to multiple universities.  Figure \ref{betavar} gives an example.  
\medskip

Although Figure \ref{betavar} primarily shows the effects on~~$\mathcal{U}_1$ when it is the less selective university (roughly corresponding to $0.5<\beta_1<0.97$), it is \textit{also} a good illustration of the interplay between \Cref{prop:opdir} and \Cref{thm:betacurve}.  When $\beta_1$ is higher or lower, $\mathcal{U}_1$ is the more selective university, so by \Cref{identmorsel}, the average type, soft skills, and quality in those portions of the graph match the single-university case.  Note that, in these sections of the graph average type and soft skill decrease (top row) but average quality increases (bottom-left) as $\beta_1$ gets closer to $\alpha=0.7$.

 \begin{figure*}[t]
     \centering
     \includegraphics[width=0.9\textwidth]{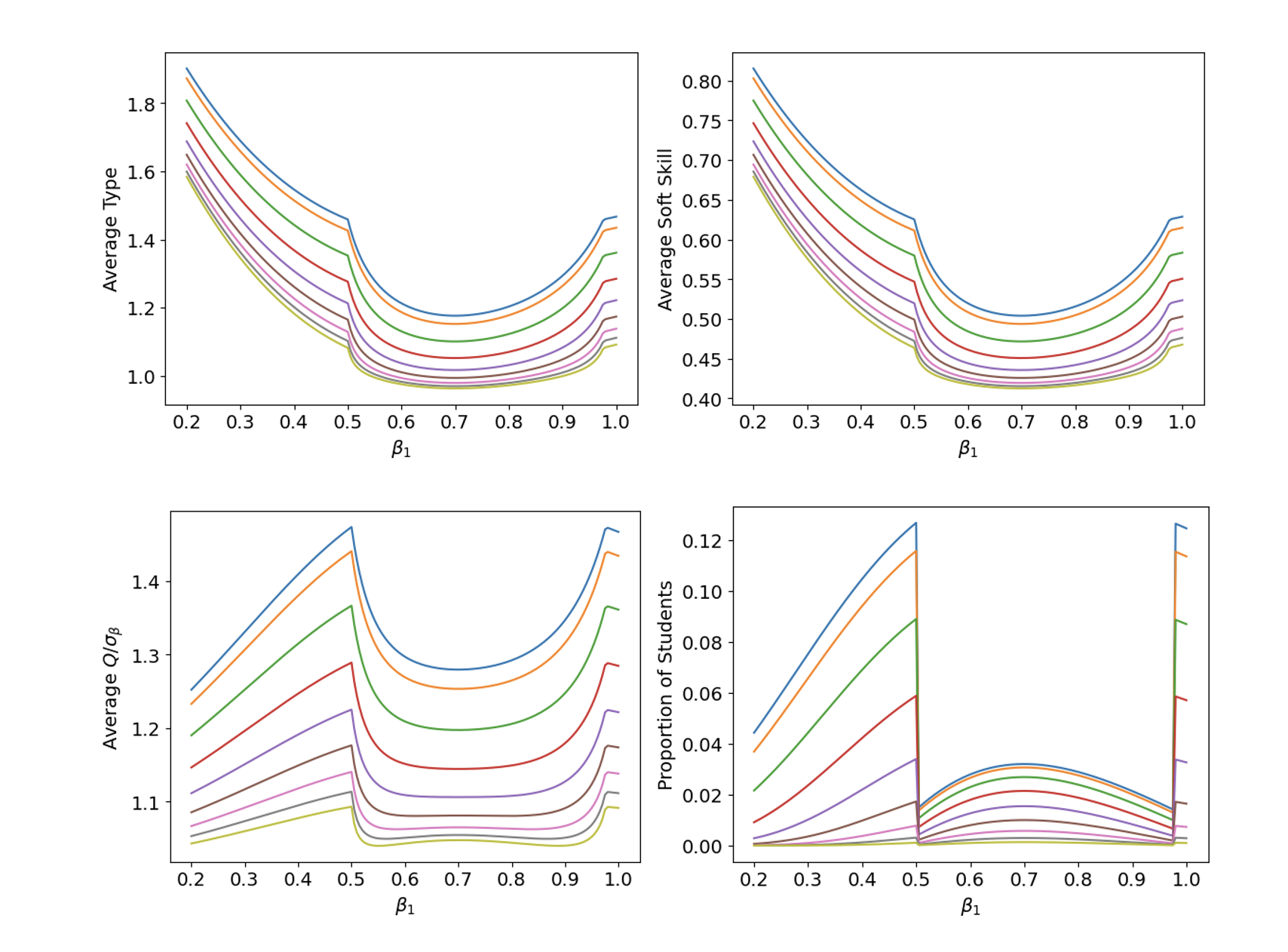}
     \caption{Parameter choice: $\alpha=0.7$, $\beta_2=0.5$, $q=0.9$, and $\tau_{Q,1}/\sigma_{\beta,1}=\tau_{Q,2}/\sigma_{\beta,2}=1$. Average type (top-left), average soft skill (top-right), average quality score (bottom-left), and total proportion of students attending (bottom-right) for University 1. The lines represent, from top to bottom, $\sigma_Z=0, 0.25, 0.5, 0.75, 1, 1.25, 1.5, 1.75, 2$.}
     \label{betavar}
 \end{figure*}

\subsection{Self-Sorting Model}

\begin{figure*}[t]
    \centering
    \includegraphics[width=0.9\textwidth]{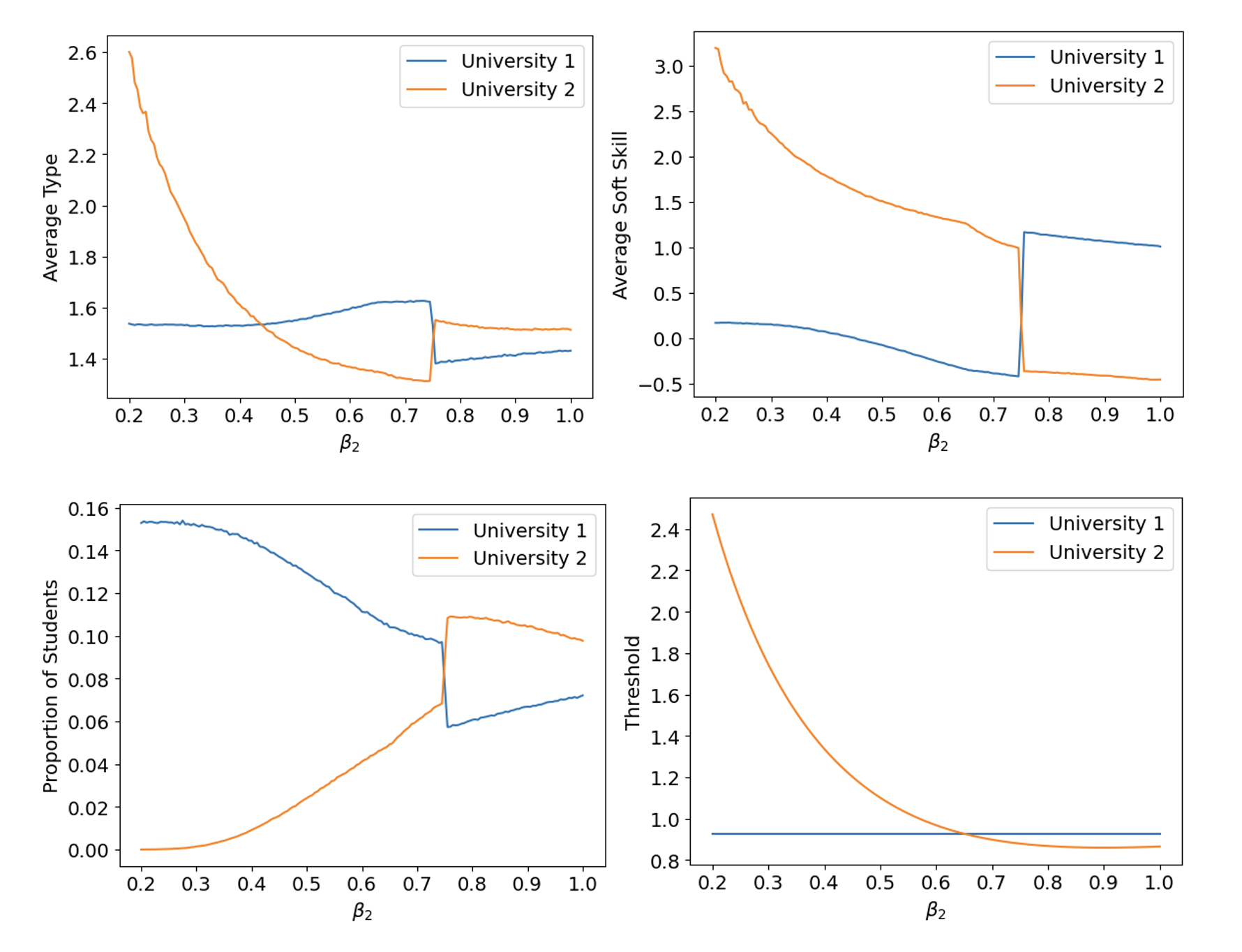}
    \caption{Parameter choice: $\alpha=0.9$, $\beta_1=0.75$, $\tau_{Q,1}/\sigma_{\beta,1}=1$, $\tau_{Q,2}/\sigma_{\beta,2}=0.95$, $\sigma_Z=0$.  Under self-sorting assumptions, the average type of attendees at each university as a function of $\beta_2$ (top-left), average soft skill of attendees at each university as a function of $\beta_2$ (top-right), proportion of students at each university as a function of $\beta_2$ (bottom-left), and threshold at each university as a function of $\beta_2$ (bottom-right).}
    \label{quadselfsort}
\end{figure*}

Here, if a student with type $t$ and soft skill $s$ is admitted to both universities, that student will choose the university $\mathcal{U}_i$ that optimizes $(\beta_it+(1-\beta_i)s)/\sigma_{\beta,i},$ where $\sigma_{\beta,i}=\sqrt{\beta_i^2+(1-\beta_i)^2}$.

Our simulations shows that insights are relatively similar to the previous case, where students admitted to both choose one university at random, with an inversion happening in which university students decide to go to. Here, remember that students chooses the university that gives them a higher quality score. This provides evidence that our results may remain robust to the specific form that our self-selection effects take. In particular, Figure~\ref{quadselfsort} provides an example where $\alpha=0.9$, $\sigma_Z=0$, $\tau_{Q,1}=1$, $\beta_1=0.75$, $\tau_{Q,2}=0.95$. This figure show the impact of varying the $\beta$ parameter of one university (here, we pick $U_2$) on the outcomes in \emph{both} universities. Figure \ref{quadselfsort}, as $\beta_2$ changes, still shows a sudden swap in the student body at each university when $\beta_2\approx 0.75=\beta_1$, as the entire segment of the student body admitted to both universities then swaps universities.

\section{Admissions through Multiple Metrics}\label{sec:mm}
We conclude with a setting where universities have access to two types of admission metrics $A_1$ and $A_2$ for each student, where $A_1 = \alpha_1 t + (1-\alpha_1)s + \sigma_{Z,1}Z_1$ and $A_2 = \alpha_2 t + (1-\alpha_2)s + \sigma_{Z,2}Z_2$. For example, the first metric $A_1$ could be standardized test scores like the SAT which are good at eliciting information about the student's type, while the second metric $A_2$ could measure extra-curricular achievements, a good indicator of the student's soft skill level. Intuitively, metrics which evaluate complementary skills (i.e., $\alpha_1$ and $\alpha_2$ are far apart), together, can provide universities a holistic overview of a student which can enable better admission decisions. Therefore, for simplicity, we specifically consider the setting with $\alpha_1=1$ and $\alpha_2 = 0$, i.e, $A_1$ offers a noisy signal about the student's type and $A_2$ offers a noisy signal about their soft-skill. We now investigate admission outcomes when universities use naive decision rules, both with and without competition. A student receives an admit from a university if and only if their score is good enough across both dimensions, i.e. iff $A_1 \geq \tau_{A_1}$ and $A_2 \geq \tau_{A_2}$. Such admission rules are quite common in practice, for example, in the case of standardized tests like the GRE or SAT, universities often impose category-wise thresholds (separate cut-offs on the verbal and quant sections of the test). 

Note that as before, there is a one-to-one mapping between the university's quality threshold and the metric-wise admission thresholds: given a desired threshold $\tau_Q$ on quality, there exist thresholds $\tau_{A,1}$ and $\tau_{A,2}$ on scores $A_1$ and $A_2$ respectively such that a student with realized scores $A_1 = a_1$ and $A_2 = a_2$ is admitted if and only if the expected student quality according to $\beta$ satisfy both: 
\begin{align}
&\mathbb{E}[\beta t+ (1-\beta)s~|~A_1 = a_1] \geq \tau_{Q},~\text{and}\\
&\mathbb{E}[\beta t+ (1-\beta)s~|~A_2 = a_2] \geq \tau_{Q}.
\end{align}

\paragraph{The single-university case} We first provide some insights about the single university setting where there is no competition. As described earlier, the university only admits those students who are expected to pass the university's threshold, conditioning on their score along each metric separately. 

\begin{proposition}
    Assume two admissions metrics, $A_1=t+\sigma_{Z,1}Z_1$ and $A_2=s+\sigma_{Z,2}Z_2$, and a single university that will accept a student with admissions metrics $(a_1,a_2)$ if $\mathbb{E}[\beta t+(1-\beta)s|A_1=a_1]\geq \tau_Q$ and $\mathbb{E}[\beta t+(1-\beta)s|A_2=a_2]\geq \tau_Q$.  Then, for fixed $\tau_Q/\sigma_{\beta}>0$, as $\beta$ increases, the average type of admitted students monotonically decreases and the average soft skill monotonically increases.
    \label{paradox}
\end{proposition}

The proof is provided in Appendix~\ref{app:paradox}. This comes from an interesting effect: suppose a university has a high $\beta$, i.e. they mostly care about type. Then, $A_1$ and $Q$ are well-aligned, but $A_2$ and $Q$ are not. In particular, among all students that have a high type, most of them will induce a high value of $\mathbb{E}[\beta t+ (1-\beta)s~|~A_1 = a_1]$ and comfortably satisfy their requirement that $\mathbb{E}[\beta t+ (1-\beta)s~|~A_1 = a_1] \geq \tau_{A,1}$. However, among these students, those that are likely to satisfy the requirement that $\mathbb{E}[\beta t+ (1-\beta)s~|~A_2 = a_2] \geq \tau_{A,2}$ must also have very high soft skill. Indeed, among these students, since $A_2$ is independent of $t$, this is the only way for them to differentiate themselves from the average and ensure they pass $\tau_{A,2}$. Formally, since $A_2 \perp t$ and $t$, admitted students must have $\mathbb{E}[s~|~A_2 = a_2] \geq \frac{\tau_{A,2}}{1-\beta}$, effectively making their bar on soft skill for admissions \emph{higher} in $\beta$. Therefore, a university should not treat the signals across two different dimensions, $A_1$ and $A_2$, separately, but instead combine them into a single score with preference parameter $\beta$.

\begin{remark}
The picture differs when a university combines signals $A_1$ and $A_2$ and makes decisions based on $\mathbb{E}[\beta t+ (1-\beta)s~|~A_1 = a_1, A_2 = a_2] \geq \tau_Q$. In particular, because $(t,A_1)$ are independent of $(s,A_2)$, one can rewrite 
\begin{align*}
&\mathbb{E}[\beta t+ (1-\beta)s~|~A_1 = a_1, A_2 = a_2] 
\\& = \beta \cdot \mathbb{E}[t|~A_1 = a_1] + (1-\beta) \cdot \mathbb{E}[s~|~A_2 = a_2]. 
\end{align*}
If $\beta$ is large, the best way for a student to be admitted by the university is to have a very high type (which is indeed the university's preferred outcome), while the dependency on soft skill is largely not taken into account. This in fact highlights the benefit of evaluating both dimensions together through a single score that aligns with the university preferences. On the contrary, as argued earlier, if the university were to filter students separately across each dimension---for example, through having a selection process where a student must separately obtain a high enough SAT score and pass an interview aimed towards soft skill---many high-type students would get filtered out due to the soft skill requirement (which is clearly not ideal from the university's point of view). 
\end{remark}

\paragraph{The competing universities case:}

Finally, we show that these effects persist in the case of multiple universities.

\begin{proposition}
    Suppose we have two universities with different values $\beta_a, \beta_b$ and thresholds $\tau_{Q,a},\tau_{Q,b}$.  Furthermore, assume, as before, that $A_1=t+\sigma_{Z,1}Z_1$ and $A_2=s+\sigma_{Z,2}Z_2$.  Assume that students admitted to both universities pick one of the universities with a uniform probability $q$.  Then there are two possibilities:

    \begin{enumerate}[label=(\alph*)]
  \item One university has a higher cutoff in both admissions metrics, and a higher average type and average soft skill.
  \item The university with the lower value of $\beta$ (i.e. it weights soft skill more heavily in its preferences) has the higher average type and the lower average soft skill.
\end{enumerate}
\label{multiparadox}
\end{proposition}
\noindent
The proof is provided in Appendix~\ref{app:multiparadox}. The insights for several universities fundamentally do not change: under competition, it is still the case that using two signals independently (one one type, one of soft skill) can lead either university to make the wrong decisions. Instead, the university should combine both signals and make joint decisions on both to recover the student quality they are expecting.

\section*{Acknowledgement}

Bentley, Sen, and Ziani were funded by NSF Awards CAREER-2336236 and IIS-2504990. This work was done in part while Ziani was visiting the Simons Institute for the Theory of Computing.

\bibliographystyle{plainnat}
\bibliography{arxiv/arxiv_ref.bib}

\appendix

\section{Proof of Claim~\ref{clm:adm_threshold}}\label{app:adm_threshold}

We only need to prove the following claim, that characterizes the expected type and soft skill of a student conditional on $A = a$.

\begin{claim}
    \label{condistri}
    Conditioned on admissions metric $A$, $\mathbb{E}[t|A = a]=\frac{\alpha a}{\sigma_{\alpha}^2}$ and $\mathbb{E}[s|A = a]=\frac{(1-\alpha) a}{\sigma_{\alpha}^2}$.
\end{claim}

\begin{proof}
First, we will calculate $\mathbb{E}[t|A]$ and $\mathbb{E}[s|A]$.  In order to calculate this, we will use the following well-known result: If $X \sim \mathcal{N}(\mu_X, \sigma_X^2)$, $Y \sim \mathcal{N}(\mu_Y, \sigma_Y^2)$, and $Cov(X, Y) = \Sigma$, then $X~|~Y = y \sim \mathcal{N}\left( \mu_X + \frac{\Sigma}{\sigma_Y^2}(y-\mu_Y), \sigma_X^2 - \frac{\Sigma^2}{\sigma_Y^2} \right)$.
In our case, since $t \sim \mathcal{N}(0, 1)$,  $s \sim \mathcal{N}(0, 1)$,  $Z \sim \mathcal{N}(0, 1)$, and all three variables are mutually independent, we have that $A \sim \mathcal{N}\left(0, \sigma_{\alpha}^2 \right)$. Further, $\text{Cov}(t, A) = \text{Cov}(t, \alpha t + (1-\alpha)s + \sigma_Z Z) = \alpha \text{Var}(t) + (1-\alpha)\text{Cov}(t,s) + \sigma_Z \text{Cov}(t,Z) = \alpha$. Plugging everything in, we obtain:
\[
    t~\vert~(A = a) \sim \mathcal{N}\left(0 + \frac{\alpha}{\sigma_{\alpha}^2}(a-0), 1-\frac{\alpha^2+\sigma_Z^2}{\sigma_{\alpha}^2} \right) \iff t ~\vert~(A = a) \sim \mathcal{N}\left( \frac{\alpha a}{\sigma_{\alpha}^2}, \frac{(1-\alpha)^2+\sigma_Z^2}{\sigma_{\alpha}^2} \right).
\]
Using an identical argument, we can also derive the desired conditional distribution of $s$. 
\end{proof}

\noindent If the university's admission threshold is $\tau_A$, the average quality of a 'borderline' admitted student must be equal to $\tau_Q$, i.e.,  
\[
    \tau_Q = \mathbb{E}[\beta t + (1-\beta)s~|~A = \tau_A] = \beta \cdot  \mathbb{E}[t~|~A = \tau_A] + (1-\beta)\cdot \mathbb{E}[s~|~A = \tau_A] = \left(\frac{\alpha \beta + (1-\alpha)(1-\beta)}{\sigma_{\alpha}^2} \right)\cdot \tau_A,
\]
where the last step follows from Claim \ref{condistri}. Re-arranging, we obtain $\tau_A = \frac{\tau_Q \sigma_{\alpha}^2}{\alpha \beta + (1-\alpha)(1-\beta)}$. This concludes the proof.

\section{Proofs for Section \ref{sec:smsu}}

\subsection{Proof of \Cref{prop:meantypesoft}}\label{app:meantypesoft}

    First, note that, by \Cref{condistri}, conditional on admissions metric $A=a$, the expected value of $t$ is $\frac{a\alpha}{\sigma_{\alpha}^2}$, and the expected value of $s$ is $\frac{a(1-\alpha)}{\sigma_{\alpha}^2}$.  Thus, if $E=\mathbb{E}[A|A\geq \tau_A]$, then $\mathbb{E}[t|A\geq \tau_A]=\frac{E\alpha}{\sigma_{\alpha}^2}$, and $\mathbb{E}[s|A\geq \tau_A]=\frac{E(1-\alpha)}{\sigma_{\alpha}^2}$.
    \medskip

    $A$ is normally distributed with mean $0$ and standard deviation $\sigma_{\alpha}$.  By a standard result, this means that $$E=\mathbb{E}[A|A\geq\tau_A]=\sigma_AH\left(\frac{\tau_A}{\sigma_A}\right).$$
    
    Thus, $\mathbb{E}[t|A\geq\tau_A]=\sigma_AH\left(\frac{\tau_A}{\sigma_A}\right)\left(\frac{\alpha}{\sigma_{\alpha}^2}\right)=\frac{\alpha}{\sigma_{\alpha}}\cdot H\left(\frac{\tau_A}{\sigma_A}\right)$.  Similarly, $\mathbb{E}[s|A\geq\tau_A]=\sigma_AH\left(\frac{\tau_A}{\sigma_A}\right)\left(\frac{1-\alpha}{\sigma_{\alpha}^2}\right)=\frac{1-\alpha}{\sigma_{\alpha}}\cdot H\left(\frac{\tau_A}{\sigma_A}\right)$.
    \medskip
    
    By Linearity of Expectation, $\mathbb{E}[Q|A\geq\tau_A]=\beta\cdot \mathbb{E}[t|A\geq\tau_A]+(1-\beta)\cdot\mathbb{E}[s|A\geq\tau_A]=\frac{\beta \alpha + (1-\beta) (1-\alpha) }{\sigma_\alpha} \cdot H(\tau_A/\sigma_{\alpha})$.

\subsection{Proof of \Cref{prop:noiselim}}\label{app:noiselim}

We will start by proving a useful monotonicity lemma.

\begin{claim}
    Suppose $x \geq 0$, and let $X$ be a random normal variable with mean $0$ and standard deviation $\sigma$.  Then, as $\sigma \rightarrow 0$, $\mathbb{E}[X\sim \mathcal{N}(0,\sigma^2)|X\geq x]$ monotonically decreases towards $x$.
\label{limbell}
\end{claim}

\begin{proof} Note that
\[
\mathbb{E}[X \mid X \geq x]
=
\sigma H\!\left(\frac{x}{\sigma}\right),
\]
where, as before,
$
H(z)=\frac{\phi(z)}{1-\Phi(z)}
$
is the Gaussian hazard rate. If $x=0$, this equals
$
\sigma H(0)=\sigma\sqrt{\frac{2}{\pi}},
$
which clearly monotonically decreases to $0$ as $\sigma\to 0$. Now suppose $x>0$. Rewriting the expression gives
\[
\mathbb{E}[X \mid X \geq x]
=
x \cdot \frac{H(x/\sigma)}{x/\sigma}.
\]
It is well-known that the function $H(z)/z$ is decreasing in $z > 0$ and converges to $1$ as $z\to\infty$. Since
$x/\sigma$ increases to $+\infty$ as $\sigma\to 0$, the result follows.
\end{proof}
As in the proof \Cref{prop:meantypesoft}, if $E=\mathbb{E}[A|A\geq\tau_A]$, and the average type is $\frac{E\alpha}{\sigma_{\alpha}^2}$.  Thus, we can set the average type as $\mathbb{E}\left[\frac{A\alpha}{\sigma_{\alpha}^2}\middle|\frac{A\alpha}{\sigma_{\alpha}^2}\geq \frac{\tau_A\alpha}{\sigma_{\alpha}^2}\right]$.  Note that, by \Cref{clm:adm_threshold}, $\tau_A = \frac{\tau_Q \sigma_{\alpha}^2 }{\alpha \beta + (1-\alpha)(1-\beta)}$, so this can be rewritten as $$\mathbb{E}\left[\frac{A\alpha}{\sigma_{\alpha}^2}\middle|\frac{A\alpha}{\sigma_{\alpha}^2}\geq \frac{\tau_Q \alpha}{\alpha \beta + (1-\alpha)(1-\beta)}\right].$$

Letting $X=\frac{A\alpha}{\sigma_{\alpha}^2}$ and $x=\frac{\tau_Q \alpha}{\alpha \beta + (1-\alpha)(1-\beta)}$, this is equivalent to $\mathbb{E}[X|X\geq x]$.  Since $A$ has mean $0$ and standard deviation $\sigma_{\alpha}$, $X$ has mean $0$ and standard deviation $\alpha/\sigma_{\alpha}$.  As $\sigma_Z$ increases, $\sigma_{\alpha}$ monotonically increases to $\infty$, so the standard deviation of $X$ monotonically decreases to $0$.  This means that, by \Cref{limbell}, the average type decreases to $x=\frac{\tau_Q \alpha}{\alpha \beta + (1-\alpha)(1-\beta)}$. 
\medskip

The proof for soft skills is nearly identical.
\subsection{Proof of \Cref{prop:alphamon}}\label{app:alphamon}

We will prove this Proposition using a quick claim.

\begin{claim}
    $\frac{\alpha}{\sigma_{\alpha}}=\frac{\alpha}{\sqrt{\alpha^2+(1-\alpha)^2+\sigma_Z^2}}$ is monotonically increasing in $\alpha$ over the interval $[0,1]$.
    \label{cosmon}
\end{claim}

\begin{proof}
    Because $\frac{\alpha}{\sigma_{\alpha}}$ is nonnegative, it suffices to show that its square, $\frac{\alpha^2}{\alpha^2+(1-\alpha)^2+\sigma_Z^2}$, is monotonically increasing.  This is easy to see, as this expression takes the form $\frac{x}{x+y+z}$, where $x$ is monotonically increasing, $y$ is monotonically decreasing, and $z$ is constant.
\end{proof}

As in the proof to \Cref{prop:noiselim}, the average type is $$\mathbb{E}\left[\frac{A\alpha}{\sigma_{\alpha}^2}\middle|\frac{A\alpha}{\sigma_{\alpha}^2}\geq \frac{\tau_Q \alpha}{\alpha \beta + (1-\alpha)(1-\beta)}\right].$$  As before, can be rewritten as $\mathbb{E}[X|X\geq x]$ letting $X=\frac{A\alpha}{\sigma_{\alpha}^2}$ and $x=\frac{\tau_Q \alpha}{\alpha \beta + (1-\alpha)(1-\beta)}$.  We will update $\mathbb{E}[X|X\geq x]$, substituting $\alpha$ for $\alpha+\epsilon$ for small $\epsilon$, in two steps, showing that neither step reduces the average type.

\begin{itemize}
    \item First, update $x$, replacing $\alpha$ with $\alpha+\epsilon$, while keeping $X$ constant.  Here, $x$ becomes $\frac{\tau_Q (\alpha+\epsilon)}{(\alpha+\epsilon)\beta_1+(1-(\alpha+\epsilon))(1-\beta_1)}$.  The numerator is linear in $\alpha$ (and positive when $\tau_Q>0$), while the denominator is a term that is linear in $\alpha$ plus a term decreasing in $\alpha$, so grows less than linearly.  Thus, $x$ increases, increasing the cutoff in the admissions metric and thus increasing the type.
    \item Second, update $X$.  Note that $X$ has mean $0$ and standard deviation $\alpha/\sigma_{\alpha}$, so with this update the standard deviation monotonically increases by \Cref{cosmon}.  Thus, the expectation $\mathbb{E}[X|X\geq x]$ increases by \Cref{limbell}.
\end{itemize}

The assertion for average soft skill follows from considering that, by symmetry, the same proof gives that soft skill monotonically increases in $1-\alpha$.

\subsection{Proof of \Cref{alphamax}}\label{app:alphamax}

Note that the average quality of admitted students can be given as $\mathbb{E}[Q|A\geq\tau_A]=\mathbb{E}\left[Q\middle|A\geq\frac{\tau_Q \sigma_{\alpha}^2 }{\alpha \beta + (1-\alpha)(1-\beta)}\right]$.  \Cref{condistri} and Linearity of Expectation give that $\mathbb{E}[Q|A = a]=a\cdot\left(\frac{\alpha \beta + (1-\alpha)(1-\beta)}{\sigma_{\alpha}^2} \right)$, so we then get that $\mathbb{E}[Q|A\geq\tau_A]=\left(\frac{\alpha \beta + (1-\alpha)(1-\beta)}{\sigma_{\alpha}^2} \right)\cdot\mathbb{E}\left[A\middle|A\geq\frac{\tau_Q \sigma_{\alpha}^2 }{\alpha \beta + (1-\alpha)(1-\beta)}\right]=\mathbb{E}\left[\frac{A(\alpha \beta + (1-\alpha)(1-\beta))}{\sigma_{\alpha}^2} \middle|\frac{A(\alpha \beta + (1-\alpha)(1-\beta))}{\sigma_{\alpha}^2}\geq\tau_Q\right]$.  Note that $\tau_Q$ is a constant, and $\frac{A(\alpha \beta + (1-\alpha)(1-\beta))}{\sigma_{\alpha}^2}$ is a Gaussian random variable with mean $0$ and standard deviation $\frac{\alpha \beta + (1-\alpha)(1-\beta)}{\sigma_{\alpha}}$.  Thus, by \Cref{limbell}, the average quality (as well as the number of students attending) is highest when this standard deviation is largest and lowest when this standard deviation is smallest.  Thus, we will analyze the derivative of this standard deviation.
\medskip

The derivative of the standard deviation can be expressed as $\frac{\sigma_{\alpha}\frac{d}{d\alpha}(\alpha\beta+(1-\alpha)(1-\beta))-(\alpha\beta+(1-\alpha)(1-\beta))\frac{d}{d\alpha}(\sigma_{\alpha})}{\sigma_{\alpha}^2}$.  The denominator is nonnegative, so it suffices to consider when the numerator is positive or negative.  Because $\sigma_{Z}=0,$ we have that $\sigma_{\alpha}=\sqrt{\alpha^2+(1-\alpha)^2}.$ Note that the numerator simplifies to

$$(2\beta-1)\sqrt{\alpha^2+(1-\alpha)^2}-\frac{(\alpha\beta+(1-\alpha)(1-\beta))(2\alpha-1)}{\sqrt{\alpha^2+(1-\alpha)^2}}.$$

We can multiply by $\sqrt{\alpha^2+(1-\alpha)^2}$, which is always positive, to get $$(2\beta-1)(\alpha^2+(1-\alpha)^2)-(\alpha\beta+(1-\alpha)(1-\beta))(2\alpha-1).$$  This in turn simplifies to $$(2\beta-1)(2\alpha^2-2\alpha+1)-(2\alpha\beta-\alpha-\beta+1)(2\alpha-1).$$  Expanding this out gives $(4\alpha^2\beta-4\alpha\beta+2\beta-2\alpha^2+2\alpha-1)-(4\alpha^2\beta-2\alpha\beta-2\alpha^2+2\alpha-2\alpha\beta+\alpha+\beta-1)=\beta-\alpha$. This derivative gives the desired conclusion.

\subsection{Proof of \Cref{prop:opdir}}\label{app:opdir}

The proof is very similar to the proof of \Cref{alphamax}.  The university will accept a student with $A=a$ iff $\mathbb{E}[Q|A=a]\geq\tau_Q$.  In order to ensure the right-hand side is constant in $\beta$, we can rewrite this as $\mathbb{E}[Q/\sigma_{\beta}|A=a]\geq\tau_Q/\sigma_{\beta}$.  Letting $E=\mathbb{E}[Q/\sigma_{\beta}|A=a]$, this means the admissions rule is thus $E\geq\tau_Q/\sigma_{\beta}$.
\medskip

Note that the average quality $\mathbb{E}[Q/\sigma_{\beta}|E\geq\tau_Q/\sigma_{\beta}]=\mathbb{E}[E|E\geq\tau_Q/\sigma_{\beta}]$.  $E$, as a linear function of $A$, is a random normal variable with mean zero, so by \Cref{limbell}, the smaller the variance of $E$, the lower the average renormalized quality.  As shown by the proof to \Cref{alphamax} (and dividing by $\sigma_\beta$ due to normalization), $E=A\cdot\left(\frac{\alpha \beta + (1-\alpha)(1-\beta)}{\sigma_{\alpha}^2\sigma_{\beta}} \right)$.  Since $A$ has a standard deviation of $\sigma_{\alpha}$, this means that the standard deviation of $E$ is $\frac{\alpha \beta + (1-\alpha)(1-\beta)}{\sigma_{\alpha}\sigma_{\beta}}$.  $\sigma_{\alpha}$ is a constant, and so it suffices to consider the behavior of the expression $\frac{\alpha \beta + (1-\alpha)(1-\beta)}{\sigma_{\beta}}$.  By similar logic to the proof of \Cref{alphamax}, this expression monotonically increases as $\beta$ increases from $0$ to $\alpha$, and monotonically decreases as $\beta$ increases from $\alpha$ to $1$. 

\subsection{Proof of \Cref{thm:betacurve}}\label{app:betacurve}

By Corollary \ref{monohaz}, it suffices to consider how
$\tau_A=\frac{\tau_Q \sigma_{\alpha}^2}{\alpha\beta+(1-\alpha)(1-\beta)}=\frac{\sigma_{\beta}(\tau_Q/\sigma_{\beta})\sigma_{\alpha}^2}{\alpha\beta+(1-\alpha)(1-\beta)}$ varies as a function of $\beta$.  Thus, we consider the derivative of $\frac{\sigma_{\beta}(\tau_Q/\sigma_{\beta})\sigma_{\alpha}^2}{\alpha\beta+(1-\alpha)(1-\beta)}$ with respect to $\beta$.  Because $\sigma_{\alpha}^2$ and $\tau_Q/\sigma_{\beta}$ are positive constants, it suffices to consider the derivative of $\frac{\sigma_{\beta}}{\alpha\beta+(1-\alpha)(1-\beta)}$.
By similar logic to the proof of \Cref{alphamax}, the derivative is zero, giving a local minimum, when $\beta=\alpha$.

\section{Proofs for Section \ref{sec:smmu}}

\subsection{Proof of \Cref{thm:more_selective}}\label{app:more_selective}

If $\tau_{Q,1}, \tau_{Q,2} > 0$, then we have by \Cref{clm:adm_threshold} that University 1 is more selective than University 2 iff $\frac{\tau_{Q,1} \sigma_{\alpha}^2}{\alpha \beta_1 + (1-\alpha)(1-\beta_1)}> \frac{\tau_{Q,2} \sigma_{\alpha}^2}{\alpha \beta_2 + (1-\alpha)(1-\beta_2)}$.  Because all constants are nonnegative and $\tau_{Q,i}$ is strictly positive, we can rearrange factors to get that University 1 is more selective iff $$\frac{\tau_{Q,1}}{\tau_{Q,2}}\cdot\frac{\alpha\beta_2+(1-\alpha)(1-\beta_2)}{\alpha\beta_1+(1-\alpha)(1-\beta_1)}>1.$$  Without loss of generality, $\beta_1>\beta_2$, so, as $\alpha$ increases from $0$ to $1$, $\frac{\alpha\beta_2+(1-\alpha)(1-\beta_2)}{\alpha\beta_1+(1-\alpha)(1-\beta_1)}$ monotonically decreases from $\frac{1-\beta_2}{1-\beta_1}$ to $\beta_2/\beta_1$.  This monotonic decrease means that, if there is any cutoff $\alpha_{th}$ at which which university is more selective changes, then University 1 is more selective whenever $0\leq \alpha<\alpha_{th}$, and University 2 is more selective whenever $\alpha_{th}<\alpha\leq 1$.

\begin{itemize}
    \item If $(\tau_{Q,1}/\tau_{Q,2})(\beta_2/\beta_1) >1\Leftrightarrow \tau_{Q,1}/\tau_{Q,2}>\beta_1/\beta_2,$ then the inequality is satisfied even at $\alpha=1$, so University 1 is more selective regardless of $\alpha$.
    \item Similarly, if $\frac{\tau_{Q,1}}{\tau_{Q,2}}\cdot\frac{1-\beta_2}{1-\beta_1}<1\Leftrightarrow\frac{\tau_{Q,1}}{\tau_{Q,2}}<\frac{1-\beta_1}{1-\beta_2}$, then the inequality is not satisfied even at $\alpha=0$, so University 2 is more selective regardless of $\alpha$.
    \item If $\frac{1-\beta_1}{1-\beta_2}\leq\frac{\tau_{Q,1}}{\tau_{Q,2}}\leq\frac{\beta_1}{\beta_2},$ then it suffices to calculate $\alpha_{th}$.  This occurs when 
    \begin{align*}
        &\frac{\tau_{Q,1}}{\tau_{Q,2}}\cdot\frac{\alpha\beta_2+(1-\alpha)(1-\beta_2)}{\alpha\beta_1+(1-\alpha)(1-\beta_1)} = 1
        \\
        &\Leftrightarrow \tau_{Q,1}(\alpha\beta_2+(1-\alpha)(1-\beta_2)) = \tau_{Q,2}(\alpha\beta_1+(1-\alpha)(1-\beta_1))
        \\
        &\Leftrightarrow 2\tau_{Q,1}\alpha\beta_2+\tau_{Q,1}-\tau_{Q,1}\beta_2-\tau_{Q,1}\alpha =2\tau_{Q,2}\alpha\beta_1+\tau_{Q,2}-\tau_{Q,2}\beta_1-\tau_{Q,2}\alpha
        \\
        &\Leftrightarrow \alpha = \frac{\tau_{Q,2}-\tau_{Q,1}+\tau_{Q,1}\beta_2-\tau_{Q,2}\beta_1}{\tau_{Q,2}-\tau_{Q,1}+2\tau_{Q,1}\beta_2-2\tau_{Q,2}\beta_1}.
    \end{align*}
\end{itemize}

\subsection{Proof of \Cref{identmorsel}}\label{app:identmorsel}

For the more selective university, this follows because the students attending the more selective university comprise a random sample of admits who each choose to attend with a constant probability $q$.
\medskip

For the less selective university, the admits are divided into two groups: the high-scoring group who are admitted to both universities, and the moderate-scoring group admitted to only one university.  The first group has a higher average type and soft skill, so when the university gets all students in the moderate-scoring group and only a proportion $1-q$ of the students in the high-scoring group, it has a lower average type and soft skill than the entire group of admits (all of whom would attend in the single university case).

\subsection{Proof of \Cref{kinv}}\label{app:kinv}

    Given University 1 having parameters $\beta_1$ and $\tau_{Q,1}$ and University 2 having parameters $\beta_2$ and $\tau_{Q,2}$, their admissions metric thresholds, per \Cref{clm:adm_threshold}, will be $\frac{\tau_{Q,1} \sigma_{\alpha}^2}{\alpha\beta_1+(1-\alpha)(1-\beta_1)}$ and $\frac{\tau_{Q,2} \sigma_{\alpha}^2}{\alpha\beta_2+(1-\alpha)(1-\beta_2)}$, respectively.  The only part of these expressions that depends on $\sigma_Z$ is $\sigma_{\alpha}$, so the ratio between the two thresholds is constant.

\subsection{Proof of \Cref{monotau}}\label{app:monotau}

Note that, as $\tau_Q$ increases by $\epsilon$, the only change in the student body is that the students with the lowest admissions metrics who did attend no longer attend.  Because expected type and soft skill are linear in admissions metric, as given by \Cref{condistri}, these students, on average, have the lowest type and soft skill.  The only exception is when a university changes from being less selective to more selective, in which case there is an instantaneous change in the number of students attending who are above both thresholds.  This change cannot reduce the average type or soft skill.

\subsection{Proof of \Cref{thm:pareto}}\label{app:pareto}

By \Cref{monotau}, average type and soft skill both increase if $\tau_{A,1}$ increases from $\tau_{A,2}-\epsilon$ to $\tau_{A,2}+\epsilon$.  The number of students attending is always continuous except at $\tau_{A,2}$, where there is a discontinuity in the proportion of students with admissions metrics greater than $\tau_{A,2}$ who attend University $1$, from $1-q$ at $\tau_{A,2}-\epsilon$ to $q$ at $\tau_{A,2}+\epsilon$.  If $q>0.5$, this leads to an increase in the number of students.  Due to this discontinuity in the number of students, and the average quality and number of students otherwise being continuous, the desired conclusion holds for any sufficiently small $\epsilon>0$.

\subsection{Proof of \Cref{prop:monorange}}\label{app:monorange}

First, note that, as follows from \Cref{corr:seldom}, the average type of the students accepted to both universities is precisely the maximum out of the average types accepted at each university.  Each average is monotonically increasing in $\alpha$ per \Cref{prop:alphamon}, so thus the average type is the maximum of two monotonically increasing functions.  This is clearly monotonically increasing.
\medskip

For the moderate-scoring group only admitted to one university, the result follows by similar logic to the proof of \Cref{prop:alphamon}, noting that the average type of this group is, similarly to \Cref{prop:alphamon},  $\mathbb{E}[X|x_2\leq X\leq x_1]$, letting $X=\frac{A\alpha}{\sigma_{\alpha}^2}$ and $x_i=\frac{\tau_{Q,i} \alpha}{\alpha \beta_i + (1-\alpha)(1-\beta_i)}$.  (We are assuming WLOG that University 1 is more selective.)
\medskip

We now show the following useful result: 

\begin{claim}
   Suppose $b>a\geq 0$, and let $Y_1$ be a random normal variable with mean $0$ and standard deviation $\sigma_1$, and let $Y_2$ be a random normal variable with mean $0$ and standard deviation $\sigma_2$.  Then, if $\sigma_2\geq\sigma_1$, $\mathbb{E}[Y_2|a\leq Y_2\leq b]\geq\mathbb{E}[Y_1|a\leq Y_1\leq b]$.
   \label{widened}
\end{claim}

\begin{proof}
   Note that it suffices to show that the ratio $f_1/f_2$ of the probability density functions $f_1$ and $f_2$ (of $Y_1$ and $Y_2$, respectively) monotonically decreases over the integral $[a,b]$.  Then we can consider the point where the two density functions, conditioned on $a\leq Y_i\leq b$, are equal, and transfer probability mass from the left of that point to the right of that point to go from $f_1$ to $f_2$, increasing the mean value in the process.  
    
   In this case, we have that $f_1(y)=C_1e^{-\frac{y^2}{2\sigma_1^2}}$, and $f_2(y)=C_2e^{-\frac{y^2}{2\sigma_2^2}}$ for some constants $C_1, C_2$.  Thus, the ratio $f_1(y)/f_2(y)$ is equal to, for some constant $C$, $Ce^{-y^2\frac{\sigma_2^2-\sigma_1^2}{2\sigma_1^2\sigma_2^2}}$.  If $\sigma_2\geq\sigma_1$, this ratio is monotonically decreasing in $y$ as desired.
\end{proof}

As in the proof of \Cref{prop:alphamon}, when increasing $\alpha$ to $\alpha+\epsilon$, we can first update $\alpha$ in the expression for $x_i$ and then update $\alpha$ in the expression for $X$.  Updating $\alpha$ in the expression for $x_i$ increases both cutoffs, following the same logic as \Cref{prop:alphamon}.  Updating $\alpha$ in the expression for $X$ increases the standard deviation of $X$ (by \Cref{cosmon}), so increases the average type by \Cref{widened}.
\medskip

Note that this proof does not necessarily work properly if which university is more selective flips.  In this case, if $\alpha_{th}$ is the point at which both universities are equally selective, the above gives that the average type in the moderate-scoring group monotonically increases for $\alpha\in(0,\alpha_{th})$ and for $\alpha\in(\alpha_{th},1)$, giving the desired conclusion.

\subsection{Proof of \Cref{discontalpha}}\label{app:discontalpha}

By \Cref{clm:adm_threshold}, the students admitted to University 1 are those with admissions metric $A\geq\frac{\tau_Q \sigma_{\alpha}^2}{\alpha\beta+(1-\alpha)(1-\beta)}$.  When $\alpha < \alpha_{th}$, then, by assumption, all such students pick University 1.  Thus, as $\alpha\rightarrow \alpha_{th}$ from the left, the limit of the average type is $\mathbb{E}\left[t\middle|A \geq \frac{\tau_{Q,1} \sigma_{\alpha}^2}{\alpha_{th} \beta_1+(1-\alpha_{th})(1-\beta_1)}\right]$.  However, as $\alpha\rightarrow \alpha_{th}$ from the right, the only students who go to University 1 are those whose admissions metric are is in the band between the two thresholds, whose width goes to $0$.  Thus, the average type converges towards $\mathbb{E}\left[t\middle|A = \frac{\tau_{Q,1} \sigma_{\alpha}^2}{\alpha_{th} \beta_1+(1-\alpha_{th})(1-\beta_1)}\right]$, which is lower.  Thus, there is a discontinuous drop.

\subsection{Proof of \Cref{prop:nonmonalpha}}\label{app:nonmonalpha}

As in the proof of \Cref{discontalpha}, when $\alpha < \alpha_{th}$, the students who pick University 1 are a random subset of the high-scoring group admitted to both universities.  Thus, as before, as $\alpha\rightarrow \alpha_{th}$ from the left, the limit of the average type is $\mathbb{E}\left[t\middle|A \geq \frac{\tau_{Q,1} \sigma_{\alpha}^2}{\alpha_{th} \beta_1 + (1-\alpha_{th})(1-\beta_1)}\right]$.  However, if $\alpha> \alpha_{th}$, the students who go to University 1 consist  of the entire moderate-scoring group admitted to one university and a proportion $\epsilon$ of the high-scoring group admitted to both universities.  As $\alpha\rightarrow \alpha_{th}$ from the right, the size of the moderate-scoring group goes to $0$ as the two admission cutoffs converge.  Thus, the proportion of students at University 1 who are in the high-scoring group goes to $1$ as $\alpha\rightarrow \alpha_{th}$, and the average type remains continuous.  However, for arbitrarily small $\delta$, we can make $\epsilon$ sufficiently small such that students in the moderate-scoring group dominate University $1$ when $\alpha= \alpha_{th} + \delta$.  This implies that the average type, as $\alpha$ increases from $\alpha_{th}$, becomes lower than when $\alpha= \alpha_{th}-\delta$, implying non-monotonicity.

\section{Proofs for Section \ref{sec:mm}}
\subsection{Proof of \Cref{paradox}}
\label{app:paradox}
    In this case, $A_1$ is independent of soft skill and $A_2$ is independent of type.  This implies three things: first, the university can determine expected type based solely on $A_1$ and expected soft skill based solely on $A_2$; second, the average soft skill conditioned on $A_1$ is zero (and vice versa); and, lastly, $A_1$ and $A_2$ are independent of each other.  
    \medskip
    
    Fixing $A_1=a_1$ and $A_2=a_2$, we can therefore apply \Cref{condistri} to get that a student's expected type is $\frac{a_1}{1+\sigma_{Z,1}^2}$, and further get that $\mathbb{E}[\beta t+(1-\beta)s|A_1=a_1]=\frac{a_1\beta}{1+\sigma_{Z,1}^2}$.  This implies that, for $A_1$, the admissions requirement is $A_1\geq \frac{\tau_Q(1+\sigma_{Z,1}^2)}{\beta}$.  Thus, because type is independent of $A_2$, the average type of admits is equal to the average type of all students with admissions metric exceeding the requirement.  Let $E=\mathbb{E}\left[A_1|A_1\geq \frac{\tau_Q(1+\sigma_{Z,1}^2)}{\beta}\right]=\mathbb{E}\left[A_1|A_1\geq \frac{(\tau_Q/\sigma_{\beta})(\sigma_{\beta})(1+\sigma_{Z,1}^2)}{\beta}\right]$.  Similarly to our argument in \Cref{prop:meantypesoft}, we have that, since expected type is linear in one's admissions metric, the expected type of those meeting the admissions requirement is $\frac{E}{1+\sigma_{Z,1}^2}$.  It is clear that increasing $\tau_{A,1}$ means that precisely the students with the lowest values for $A_1$ are eliminated, so the average type increases.  Thus, we need to show that, as $\beta$ decreases, $\tau_{A,1}$ increases.  This holds because $\tau_{A,1}=\frac{(\tau_Q/\sigma_{\beta})(\sigma_{\beta})(1+\sigma_{Z,1}^2)}{\beta}$, $\tau_Q/\sigma_{\beta}$ and $\sigma_Z$ are constants, and a similar argument to \Cref{cosmon} gives that $\sigma_\beta/\beta$ increases as $\beta$ decreases.  This means that average type increases if the university focuses more on soft skill rather than type.
    \medskip

    A similar argument examining $A_2$ shows that average soft skill monotonically increase as $\beta$ increases.

\subsection{Proof of \Cref{multiparadox}}
\label{app:multiparadox}
    \Cref{paradox} gives that the students admitted to University $a$ are those with $A_1\geq \frac{\tau_{Q,a}(1+\sigma_{Z,1}^2)}{\beta_a}$ and $A_2\geq \frac{\tau_{Q,a}(1+\sigma_{Z,2}^2)}{1-\beta_a}$.  Similarly, the students admitted to University $b$ are those with $A_1\geq \frac{\tau_{Q,b}(1+\sigma_{Z,1}^2)}{\beta_b}$ and $A_2\geq \frac{\tau_{Q,b}(1+\sigma_{Z,2}^2)}{1-\beta_b}$.  WLOG, assume that University $a$ has a higher cutoff for $A_1$; i.e.,  $\frac{\tau_{Q,a}(1+\sigma_{Z,1}^2)}{\beta_a}\geq \frac{\tau_{Q,b}(1+\sigma_{Z,1}^2)}{\beta_b}$.  There are two cases:

    \begin{itemize}
        \item If University $a$ also has a higher cutoff for $A_2$, all students accepted to University $a$ must also meet the cutoffs for University $b$.  Thus, the set of students attending University $a$ is a random sample of the students accepted to both universities.  Because the students at University $a$ meet higher cutoffs on both metrics, their average type and average soft skill must be higher than the group admitted to University $b$.  The remaining admits to University $b$, who attend University $b$, must then have a lower average type and soft skills.  This corresponds to case (a) in the proposition statement.
        \item If University $b$ has a lower cutoff for $A_2$, there are three groups we must consider:
        \begin{itemize}
            \item Students admitted to both universities choose a university at random.
            \item After considering students in the first group, students admitted only to University $a$ were not admitted to University $b$ because their predicted soft skill based on $A_2$ was too low.  These students lower the average soft skill at University $a$ compared to those in the first group, but do not change the average type.
            \item Students admitted only to University $b$ were not admitted to University $a$ because their predicted type based on $A_1$ was too low.  These students lower the average type at University $b$ compared to those in the first group, but do not change the average soft skill.
        \end{itemize}
        Moreover, in this case, we have that $\frac{\tau_{Q,a}(1+\sigma_{Z,1}^2)}{\beta_a}\geq \frac{\tau_{Q,b}(1+\sigma_{Z,1}^2)}{\beta_b}$, but $\frac{\tau_{Q,a}(1+\sigma_{Z,2}^2)}{1-\beta_a}<\frac{\tau_{Q,b}(1+\sigma_{Z,2}^2)}{1-\beta_b}$.  This implies that $\left(\frac{\tau_{Q,a}(1+\sigma_{Z,1}^2)}{\beta_a}\right)\div\left(\frac{\tau_{Q,a}(1+\sigma_{Z,2}^2)}{1-\beta_a}\right)>\left(\frac{\tau_{Q,b}(1+\sigma_{Z,1}^2)}{\beta_b}\right)\div\left(\frac{\tau_{Q,b}(1+\sigma_{Z,2}^2)}{1-\beta_b}\right)\Rightarrow \frac{1+\sigma_{Z,1}^2}{1+\sigma_{Z,2}^2}\frac{1-\beta_a}{\beta_a}> \frac{1+\sigma_{Z,1}^2}{1+\sigma_{Z,2}^2}\frac{1-\beta_b}{\beta_b}\Rightarrow\beta_a<\beta_b$.  This is case (b) of the proposition statement.
    \end{itemize}

\end{document}